\documentclass[final,13pt,times]{elsarticle}

\usepackage{algorithm}
\usepackage{algpseudocode}
\usepackage[utf8]{inputenc}
\usepackage{graphicx}
\usepackage{url}
\usepackage{comment}

\usepackage{amsthm}
\usepackage{amsmath,amssymb}
\usepackage{draftwatermark}
\usepackage{bbm}
\usepackage{blindtext}
\usepackage{hyperref}
\usepackage{multicol}
\usepackage{orcidlink}

\newtheorem{theorem}{Theorem}
\newtheorem{lemma}[theorem]{Lemma}
\newtheorem{corollary}{Corollary} 
\newcommand{\HSIC}{\mathrm{HSIC}}

\renewenvironment{thebibliography}[1]{
  \begin{oldthebibliography}{#1}
    \setlength{\itemsep}{0em}
    \setlength{\parskip}{0em}
}
{
  \end{oldthebibliography}
}

\begin{document}

\begin{frontmatter}
\title{Testing multivariate normality by testing independence}
\author[label1]{Povilas Daniu\v{s}is, \orcidlink{0000-0001-5977-827X}}

\affiliation[label1]{organization={Vytautas Magnus University},
            addressline={Research Institute of Natural and Technological Sciences}, 
            city={Akademija, Kaunas distr.},
            postcode={53361},         
            country={Lithuania}}



\begin{abstract}
We propose a simple multivariate normality test based on Kac-Bernstein's characterization, which can be conducted by utilising existing statistical independence tests for sums and differences of data samples. We also perform its empirical investigation, which reveals that for high-dimensional data, the proposed approach may be more efficient than the alternative ones.  The accompanying code repository is provided at \url{https://shorturl.at/rtuy5}.

\end{abstract}

\end{frontmatter}

\section{Introduction}
Normality assessment is a classical problem in statistics. Since marginal normality, in general, does not imply joint normality, multivariate normality tests cannot be conducted by performing a series of component-wise univariate normality testing (e.g., perturbed bivariate normal density 
$p(x,y) = \frac{1}{2\pi} e^{-\frac{1}{2} (x^2 + y^2)}(1 + xy e^{-\frac{1}{2}(x^2 + y^2)})$   is non-normal, however marginals $p(x)$ and $p(y)$ are normal). Similarly to the univariate case, these tests are usually constructed by directly investigating asymptotics of various multivariate statistics~\cite{ebner}. In this article, we attempt to take a different point of view. We rely on the equivalent normality characterization in terms of statistical independence, and utilise the existing statistical independence tests as the basis of our proposed multivariate normality testing procedure. We begin with the review of the related previous work in Section~\ref{sec:previous_work}. Afterward, we formulate the proposed independence-based multivariate normality test in Section~\ref{sec:proposed_method}. In Section~\ref{sec:simulations}, we conduct an empirical analysis of this test using a set of multivariate probability distributions, comparing it with the previous work. Finally, we conclude the article with Section~\ref{sec:conclusion}. The primary our contributions  are a.)
the first algorithmic application of the Kac-Bernstein theorem for assessing multivariate normality through statistical independence testing and b.) an empirical evidence showcasing the effectiveness of the proposed approach, comparing to Henze-Zirkler test~\cite{Henze}. 


\section{Previous work}
\label{sec:previous_work}
Let us briefly review the related previous work in multivariate normality and statistical independence testing.

\textbf{Multivariate normality tests}
Let us denote by $\mathcal{N}$ a normal distribution, and let $X$ be $d$-dimensional random vector, defined in some probability space $(\Omega_{X}, \mathcal{F}_{X}, \mathcal{P}_{X})$. Let $X^{1:n} := (x_{1},x_{2},...,x_{n})$ be data sample, consisting of $n$ i.i.d. realisations of $X$. Various approaches have been proposed to test the statistical hypothesis (e.g.~\cite{ebner, Henze, Arnastauskaite}): 

\begin{equation}
\label{eq:main_problem}
H_{0}: \mathcal{P}_{X} = \mathcal{N} \textnormal{ vs. }  H_{1}: \mathcal{P}_{X} \neq \mathcal{N}.
\end{equation}

\textbf{Henze-Zirkler (HZ) multivariate normality test} In our experiments (Section~\ref{sec:simulations}) we rely on Henze-Zirkler (HZ) ~\cite{Henze} test as a baseline.
Let $\bar{x}$ be the sample mean, and let $S := \frac{1}{n}\sum_{i=1}^{n}(x_{i} - \bar{x})(x_{i} - \bar{x})^{T}$ be empirical covariance. Let us denote $D_{i,j} := (x_{i} - x_{j})^{T}S^{-1}(x_{i} - x_{j})$, and $D_{i} := (x_{i} - \bar{x})^{T}S^{-1}(x_{i} - \bar{x})$. The Henze-Zirkler test statistic is given by~\cite{Henze}
\begin{equation}
\label{eq:hzstat}
HZ(h) := \frac{1}{n} \sum_{i,j=1}^{n} e^{-\frac{h^2}{2} D_{i,j}} - 2(1+h^2)^{-\frac{d}{2}}\sum_{i=1}^{n} e^{-\frac{h^2}{2(1+h^2)} D_{i}} + n(1+2h^{2})^{-\frac{d}{2}}.
\end{equation}

\noindent If the sample is normally distributed, Eq.~\eqref{eq:hzstat} is approximated by the log-normal distribution with the mean $1 - \frac{a^{-\frac{d}{2}}(1 + dh^{\frac{2}{a}} + d(d+2)h^4)}{2a^2}$, and the variance $2(1+4h^2)^{-\frac{d}{2}} + \frac{2a^{-d}(1+2dh^4)}{a^2} + \frac{3d(d+2)h^8}{4a^4} - 4w_{h}^{-\frac{d}{2}} (1 + \frac{3dh^4}{2w_{h}} + \frac{d(d+2)h^8}{2w_{h}^2})$,
where $a = 1+2h^{2}$, and $w_{h} = (1+h^2)(1+3h^3)$ (see~\cite{Henze}).
Henze and Zirkler recommend the following choice of $h$ by $h^{*} = \frac{1}{\sqrt{2}}\left(\frac{n(2d+1)}{4}\right)^{\frac{1}{d+4}}$ (Section 2 of~\cite{Henze}).



\noindent \textbf{Statistical independence tests} 
Let $(X,Y)$ be a random vector, defined in probability space $(\Omega_{X,Y}, \mathcal{F}_{X,Y}, \mathcal{P}_{X,Y})$, and let $(X^{1:n};Y^{1:n}) := ( (x_{1}, y_{1}),...,(x_{n},y_{n}) )$ be a sample, consisting of $n$ i.i.d. it's realisations. Let us consider statistical hypothesis:
\begin{equation}
\label{eq:independence_hypothesis}
H_{0}: \mathcal{P}_{X,Y} = \mathcal{P}_{X}\mathcal{P}_{Y} \textnormal{ vs. }  H_{1}: \mathcal{P}_{X,Y} \neq \mathcal{P}_{X}\mathcal{P}_{Y}.
\end{equation}
\noindent \textbf{HSIC-based statistical independence test} In order to test the hypothesis $H_0$ against the alternative $H_1$ (Eq.~\eqref{eq:independence_hypothesis}) we will further utilise kernel-based test~\cite{hsic}. Having two positive definite kernels $k(.,.)$ and $l(.,.)$~\cite{learning_with_kernels}, this test relies on the asymptotics of HSIC dependence measure estimator~\cite{hsic}:

\begin{equation}
\label{eq:hsic_estimator}
    \HSIC_{b}(X^{1:n};Y^{1:n}) := \frac{1}{n^2}Tr(KHLH),
\end{equation}

\noindent where $Tr(.)$ is matrix trace operator, $K$ and $L$ are corresponding Gram matrices with entries $K_{i,j}=k(x_{i},x_{j})$, and $L_{i,j}=l(y_{i},y_{j})$, respectively, and $H = I - n^{-1}11^{T}$ is the centering matrix.
In the case of the null hypothesis (Eq.~\eqref{eq:independence_hypothesis}), the distribution of $n \HSIC_{b}(X^{1:n};Y^{1:n})$ is well approximated by the gamma distribution with probability density $p_{gamma}(x) = \frac{x^{u - 1 } e^{-x/v}}{v^{u}\Gamma(v)} $, where parameters $u = \frac{(\mathbb{E} \HSIC_{b}(Z))^2}{Var \HSIC_{b}(Z)}$ and $v = \frac{n Var \HSIC_{b}(Z)}{ \mathbb{E} \HSIC_{b}(Z)}$ are defined only in terms of $K$ and $L$ (see~\cite{hsic} for details).

Therefore, given the sample $(X^{1:n};Y^{1:n})$, and a specified type I error probability $\alpha$, the $H_{0}$ in ~\eqref{eq:independence_hypothesis} is accepted when $n \HSIC_{b}(X^{1:n};Y^{1:n}) \leq Q_{1-\alpha}(X^{1:n};Y^{1:n})$, where $Q_{1-\alpha}(X^{1:n};Y^{1:n})$ represents the $1-\alpha$-quantile of corresponding gamma distribution. Let us define the corresponding indicator variable:

\begin{equation}
\label{eq:statistical_independence_test}
T_{\alpha}(X^{1:n};Y^{1:n}) :=   
\begin{cases}
  0 & \text{if } n \HSIC_{b}(X^{1:n};Y^{1:n}) \leq Q_{1-\alpha}(X^{1:n};Y^{1:n}), \\
  1 & \text{otherwise}.
  \end{cases}
\end{equation}


\section{Proposed multivariate normality test}
\label{sec:proposed_method}
In~\cite{Valderrama} authors provide the multivariate normality test, relying on statistical independence assessment. However, their work bases on Darmois-Skitovich theorem~\cite{skitovivc1962linear}, which provides only sufficient normality conditions. Moreover, their method requires performing the principal component analysis (PCA,~\cite{pearson1901liii}), and testing for their mutual independence, among other conditions, which limits its applicability for the higher-dimensional data. On the contrary, we apply Kac-Bernstein theorem~\cite{Kac, ref1}, which characterizes normal distribution in terms of the independence structure of four random vectors. Let us further denote the independence of random vectors $X$ and $Y$ by $X \perp Y$.

\begin{theorem}[Kac-Bernstein]
\label{thm:KacBernstein}
Let $X \perp Y$ be two random vectors. If $X-Y \perp X + Y$, then  $X$ and $Y$ are normal.
\end{theorem}

\noindent We will utilise the following its variation:

\begin{corollary}
\label{thm:KacBernstein1}
Let $X \perp Y$ be two $d$-dimensional random vectors with the same probability distributions. Then $X-Y \perp X + Y$ if and only if $X$ and $Y$ are normal.
\end{corollary}
\begin{proof}
The implication trivially follows from Theorem~\ref{thm:KacBernstein}. On the other side, let $X$ and $Y$  $\sim \mathcal{N}(\mu, \Sigma)$ be the random vectors with coinciding characteristic functions $\phi(\gamma) = e^{i \gamma^{T} \mu - \frac{1}{2} \gamma^{T}\Sigma \gamma}$, where $i^{2} = -1$, $\mu \in \mathbb{R}^d$ and $\Sigma \in \mathbb{R}^{d \times d}$ are the mean and covariance matrix. Then, because both $X$ and $Y$ are independent normal random vectors, $(X-Y,X+Y)$ are jointly normal, since the characteristic function 
\begin{equation}
\label{eq:cf}
    \phi_{X-Y,X+Y}(\alpha,\beta) := \mathbb{E}e^{i\alpha^{T}(X-Y)+i\beta^{T}(X+Y)} = \phi(\beta+\alpha)\phi(\beta - \alpha),
\end{equation}
\noindent where $\alpha,\beta \in \mathbb{R}^{d}$, is characteristic function of a normal random vector. Because of joint-normality, the independence of $X-Y$ and $X + Y$ follows from uncorelatedness, and $\Sigma_{X-Y,X+Y} = \mathbb{E} (X-Y)(X+Y)^{T} - \mathbb{E} (X-Y) \mathbb{E} (X+Y)^{T}$ is zero because of i.i.d. condition.
\end{proof}

\noindent Let $X^{1}$ and $X^{2}$ be two independent random vectors distributed according to possibly unknown probability distribution $P_{X}$. Since they satisfy conditions of Corollary~\ref{thm:KacBernstein1}, statistical hypothesis, defined in Eq.~\ref{eq:main_problem} can be  reformulated as

\begin{equation}
\label{eq:dual_problem}
H_{0}: X^{1}-X^{2} \perp X^{1} + X^{2}  \textnormal{ vs. } H_{1}:  X^{1}-X^{2} \not \perp X^{1} + X^{2},
\end{equation}

\noindent which can be equivalently tested using statistical independence tests (e.g.~\cite{hsic, Szekely}).  

\begin{algorithm}
\caption{Kac-Bernstein normality test}\label{alg:cap}
\begin{algorithmic}
\Require Statistical independence test $T_{\alpha}$, with significance level $\alpha \in (0,1)$, and i.i.d. samples $x_{1},...,x_{2n}$, samples from the distribution of interest.
\State Denote $X^{1,1:n} := x_{1},...,x_{n}$, and $X^{2,1:n} := x_{n+1},...,.x_{2n}$.
\State Return $T_{\alpha}(X^{1,1:n} - X^{2,1:n}, X^{1,1:n}  + X^{2,1:n})$.

\end{algorithmic}
\label{algo:kbtest_algo}
\end{algorithm}

Because of the equivalence relation between normality and independence in Corollary~\ref{thm:KacBernstein1}, the characteristics (e.g. critical values) of the proposed test (Algorithm~\ref{algo:kbtest_algo}) corresponds to that of given statistical independence test $T_{\alpha}$.

\section{Simulations}
\label{sec:simulations}
Let us empirically compare the proposed independence-based multivariate normality test (Algorithm~\ref{algo:kbtest_algo}) with Henze-Zirkler multivariate normality test using samples from a set of different normal and non-normal distributions of different sizes $n \in \{ 1000,2000,3000,4000,5000 \}$ and dimensionalities $d \in [50,150]$.

As the component of Algorithm~\ref{algo:kbtest_algo} we utilise HSIC-based~\cite{hsic} statistical independence test with a significance level set to $0.05$. In this test we utilise Gaussian kernel $k(x,y) = e^{-\frac{||x-y||^2}{2\sigma^{2}}}$, with bandwidth parameter $\sigma$ set using the median heuristics~\cite{Garreau2017LargeSA}.

\paragraph{Normal distributions} We analyse multivariate normal distributions:

$$\mathcal{N}(0,I), \mathcal{N}(0,C),\mathcal{N}(\mu,I),\mathcal{N}(\mu,C),$$ where $\mu$ is uniform random vector with components in $[-1,1]$, $C=UU^{T}$, and $U$ is random matrix with uniform entries. We define empirical significance level (see Eq. ~\eqref{eq:main_problem}):

\begin{equation}
\label{eq:empirical_power_normal}
\widehat{P(H_{1}|H_0)} := \frac{1}{n_{E}} \sum_{i=1}^{n_{E}} T_{\alpha}(X_{i}^{1,1:n}- X_{i}^{2,1:n}, X_{i}^{1,1:n}+ X_{i}^{2,1:n}),
\end{equation}
\noindent where $n_{E}$ is number of samples, and $X_{i}^{1,1:n}$, $X_{2}^{2,1:n}$ are as in Algorithm~\ref{algo:kbtest_algo} and have normal distribution.

\paragraph{Non-normal distributions} In this case, we analyse multivariate data with independent components distributed according to $13$ distributions, encompassing both symmetrical and asymmetrical ones: 

$$\chi^{2}(1),\chi^{2}(2), U[0,1], U[-1,1],Laplace(0,1),Logistic(0,1),  Logistic(0,2),$$
$$Power(2), Cauchy(0,1), Beta(5,5), Beta(8,2),$$ 
$$Beta(2,8), 0.5\mathcal{N}(0,1) + 0.5\mathcal{N}(0.5, 1).$$ 

\noindent We define empirical power:

\begin{equation}
\label{eq:empirical_power_normal}
\widehat{P(H_{1}|H_{1})} := \frac{1}{n_{E}} \sum_{i=1}^{n_{E}} T_{\alpha}(X_{i}^{1,1:n}- X_{i}^{2,1:n}, X_{i}^{1,1:n}+ X_{i}^{2,1:n}),
\end{equation}
\noindent where $n_{E}$ is number of samples, and $X_{i}^{1,1:n}$, $X_{2}^{2,1:n}$ are as in Algorithm~\ref{algo:kbtest_algo} and have non-normal distribution.

\begin{figure}[h]
    \centering
    \includegraphics[width=0.80\textwidth]{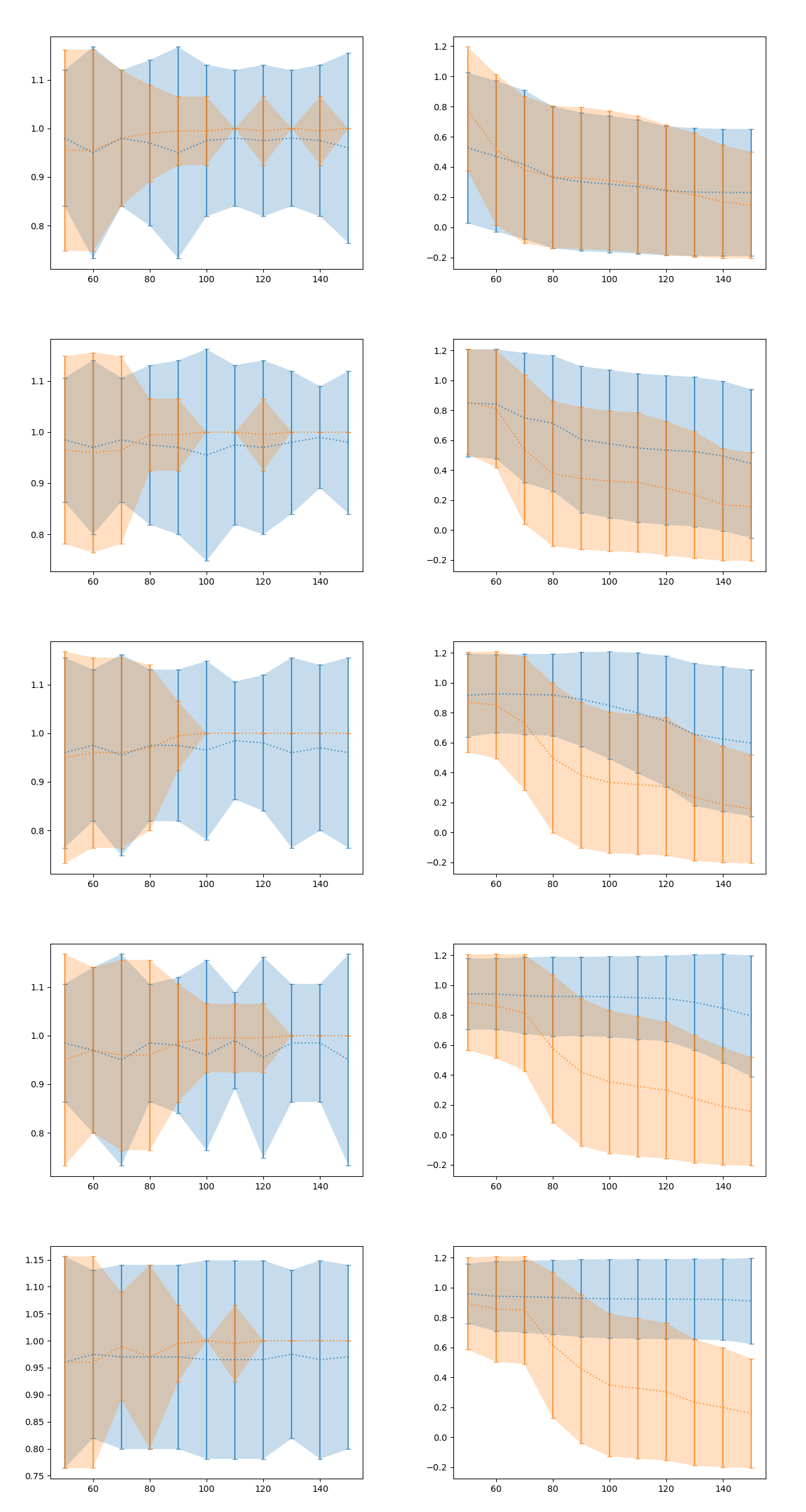}
    \caption{Left: data dimension ($x$ axis) versus average $1 - \widehat{P(H_{1}|H_{0})}$ ($y$ axis) for normal data. Right: data dimension ($x$ axis) versus average $\widehat{P(H_{1}|H_{1})}$ ($y$ axis) for non-normal data.    
    Orange graphs correspond to the Henze-Zirkler test, and the blue ones to the proposed normality test (Algorithm~\ref{algo:kbtest_algo}). From top to bottom, the graphs are for different data sizes $n \in \{ 1000,2000,3000,4000,5000 \}$.}
    \label{fig:empirical_power}
\end{figure}

In order to assess statistical significance, we repeat the experiments for each dimensionality $n_{E} = 50$ times providing corresponding average test performances and standard deviations in Fig.~\ref{fig:empirical_power}, averaged according to all distributions. Fig.~\ref{fig:empirical_power} show that for larger data cases, the average performance of the Henze-Zirkler test declines, while that of the suggested independence-based approach (Algorithm ~\ref{algo:kbtest_algo}) is more efficient with the increasing dimensionality of data. We also conducted experiments with smaller dimensionalities $d < 50$, in which the aforementioned advantage of the suggested independence-based approach was not observed.


\section{Conclusion}
\label{sec:conclusion}
In this short study, we proposed a new multivariate normality test, based on Kac-Bernstein characterization theorem, and provided its empirical investigation using a set of 
multivariate probability distributions. A notable property of our test is its equivalence to statistical independence assessment. Thereby, with the improving statistical independence testing methods, one can directly obtain more efficient multivariate normality tests.
Our experiments reveal that for larger dimensions our testing procedure was more efficient than Henze-Zirkler test, for larger data samples (e.g. $n \geq 2000$), although for lower sample sizes (e.g. $n = 1000$) both tests performed similarly.  
Besides more detailed empirical investigation, future work may include the construction of multivariate normality measures, which rely on the statistical dependence measure. For example, having two $i.i.d.$ random vectors $X$ and $X'$, and some statistical dependence measure $Dep(X,X')$~\cite{hsic, daniusis2022measuring}, one may estimate normality of their distribution via $Dep(X - X',X +X')$. Such normality measures may be useful e.g. in maximum entropy regularization and causal inference~\cite{JMLR:v17:14-375}.


\section{Acknowledgements}
We thank  Neurotechnology and Vytautas Magnus University for partial support, professors Vilijandas Bagdonavi\v{c}ius, Virginijus Marcinkevi\v{c}ius, Pranas Vaitkus, and other colleagues for useful comments.

\bibliographystyle{unsrt}

{\footnotesize
\bibliography{bibliography}
}

\end{document}